\documentclass[11pt,a4paper]{article}
\usepackage[T1]{fontenc}
\usepackage{lmodern}
\usepackage[margin=27mm]{geometry}
\usepackage{amsmath,amssymb,amsthm,mathtools}
\usepackage{graphicx}
\usepackage[numbers,sort&compress]{natbib}
\usepackage{microtype}
\usepackage{caption}
\usepackage[hidelinks]{hyperref}
\hypersetup{
  pdftitle={Symmetry Transfer to Basins and Enhancement at Basin Boundaries},
  pdfauthor={Xiao Xie},
  pdfsubject={Equivariant dynamics and basin-boundary symmetry}}
\newtheorem{theorem}{Theorem}
\newtheorem{lemma}[theorem]{Lemma}
\newtheorem{corollary}[theorem]{Corollary}
\newtheorem{proposition}[theorem]{Proposition}
\newcommand{\R}{\mathbb R}
\newcommand{\C}{\mathbb C}
\newcommand{\N}{\mathbb N}
\DeclareMathOperator{\dist}{dist}
\DeclareMathOperator{\Fix}{Fix}
\DeclareMathOperator{\Adm}{Adm}
\title{Symmetry Transfer to Basins and Enhancement at Basin Boundaries}
\author{Xiao Xie\\[2pt]
\small School of Physics and Astronomy, Sun Yat-sen University\\
\small\texttt{xiex68@mail2.sysu.edu.cn}}
\date{}

\begin{document}
\maketitle

\begin{abstract}
For an equivariant dynamical system, we prove that a compact Lyapunov-stable attractor realized as an $\omega$-limit set has the same setwise symmetry group as its basin. Admissible attractor and basin symmetries therefore coincide within each fixed dynamical category. For a finite symmetry group, the index of the attractor stabilizer in the boundary stabilizer counts the symmetry-related basins sharing that full boundary. This yields a criterion for boundary symmetry enhancement, a dichotomy for attractor orbits of prime size, and a connection with Wada basins. A common boundary of $m$ disjoint open basins has at least $m$ local complementary components, which excludes smooth hypersurface patches when $m\geq3$. Smooth flows and Newton maps illustrate the different possibilities.
\end{abstract}

\section{Introduction}
Symmetry organizes coexisting attractors into group orbits and acts on the initial conditions that approach them \cite{GS2002,Field2007}. The resulting basin geometry is central to the description of multistability. Work on conditional symmetry has also identified symmetric pairs of attracting states with asymmetric basin configurations \cite{Li2017,Li2023}. These examples motivate a precise treatment of how a symmetry acts throughout the basin. We consider a fixed group action commuting with the dynamics on the entire phase space. For a smooth vector field $f$ and a smooth transformation $g$, this condition takes the form $f(gx)=Dg_x f(x)$.

The admissible-subgroup problem asks which subgroups of a prescribed symmetry group occur as exact setwise stabilizers of attractors. Melbourne, Dellnitz, and Golubitsky \cite{MDG1993} and Ashwin and Melbourne \cite{AM1994} developed the theory for continuous equivariant maps; Golubitsky and Stewart \cite[Chapter~9]{GS2002} give a systematic account. Diffeomorphisms and flows have their own realization theory \cite{FMN1996}. These results concern compact attractors with specified stability and recurrence properties. Their generally noncompact basins require a link between the two stabilizers.

For the attractor class considered here, we establish
\begin{equation}\label{eq:hierarchy}
G_A=G_{B(A)}\leq G_{\partial B(A)}.
\end{equation}
The equality transfers admissibility to basins in the same dynamical category. The remaining inclusion is strict exactly when a symmetry carries $A$ to a distinct attractor with the same full basin boundary. Writing $H=G_A$ and $K=G_{\partial B(A)}$, the boundary assignment on the attractor orbit is the projection $G/H\to G/K$. Each fiber contains $[K:H]$ basins. For open basins with nonempty boundary, fibers of size at least three are Wada families in the common-boundary sense \cite{KY1991,Daza2018}. We derive the associated local separation constraint and illustrate multiplicities one, two, and arbitrary integers at least three.

\section{Attractor--basin symmetry transfer}\label{sec:transfer}
Let $(X,d)$ be a metric space and let a finite group $G$ act by isometries. Write $\Phi_t$ for a continuous semiflow, $t\in\R_{\geq0}$, or the iterates of a continuous map, $t\in\N_0$, satisfying
\begin{equation}\label{eq:equivariance}
\Phi_t(gx)=g\Phi_t(x),\qquad g\in G,\ x\in X.
\end{equation}
All time parameters below belong to the corresponding time set. For $S\subset X$, its setwise stabilizer is
\[
G_S=\{g\in G:gS=S\}.
\]
All boundaries are taken in $X$, and all stabilizers are taken within the specified group $G$.

An attractor $A$ is a nonempty compact forward-invariant set that is Lyapunov stable and satisfies $A=\omega(x_*)$ for some $x_*\in X$. These conditions follow \cite[Definition~9.1]{GS2002} for maps and are used here also for semiflows. The limit set is
\[
\omega(x)=\{y\in X:\Phi_{t_j}(x)\to y\text{ for some }t_j\to\infty\}.
\]
Lyapunov stability means that each open $U\supset A$ contains an open $V\supset A$ such that $\Phi_t(V)\subset U$ for all $t\geq0$. Define
\begin{equation}\label{eq:basin}
B(A)=\{x\in X:\dist(\Phi_t(x),A)\to0\text{ as }t\to\infty\}.
\end{equation}
For the Wada results in Section~\ref{sec:boundary}, we additionally require $A$ to attract an open neighborhood.

\begin{lemma}\label{lem:basin}
For an attractor $A$ as above,
\begin{equation}\label{eq:omega-basin}
B(A)=\{x:\varnothing\ne\omega(x)\subseteq A\}
=\{x:\omega(x)\cap A\ne\varnothing\}.
\end{equation}
Distinct attractors in this class are disjoint and have disjoint basins. If $A$ attracts a neighborhood, then $B(A)$ is open.
\end{lemma}
\begin{proof}
Convergence in distance to compact $A$ gives a subsequential limit in $A$, and every $\omega$-limit point lies in $A$. Conversely, let $p\in\omega(x)\cap A$. Given $\varepsilon>0$, stability provides an open neighborhood $V$ of $A$ whose forward images stay within distance $\varepsilon$ of $A$. The orbit enters $V$ along a subsequence converging to $p$ and remains within that distance thereafter. This proves \eqref{eq:omega-basin}.

If $p\in A\cap C$ for two such attractors and $A=\omega(x_*)$, the preceding argument gives $x_*\in B(C)$, hence $A\subseteq C$. Interchanging $A$ and $C$ yields equality. Distinct compact disjoint sets have positive distance, so their basins are disjoint. Finally, for an open neighborhood $U$ of $A$ contained in $B(A)$, compactness of $A$ implies that every basin orbit eventually enters $U$. Thus $B(A)=\bigcup_{t\geq0}\Phi_t^{-1}(U)$ is open.
\end{proof}

\begin{theorem}[Symmetry transfer]\label{thm:transfer}
For every $g\in G$,
\begin{equation}\label{eq:covariance}
\omega(gx)=g\omega(x),\qquad
 gB(A)=B(gA),\qquad
 g\partial B(A)=\partial B(gA).
\end{equation}
Moreover, $G_A=G_{B(A)}\leq G_{\partial B(A)}$.
\end{theorem}
\begin{proof}
Equivariance and continuity give $g\omega(x)\subseteq\omega(gx)$, and applying $g^{-1}$ yields equality. Since $g$ is an isometry,
\[
\dist(\Phi_t(gx),gA)=\dist(\Phi_t(x),A),
\]
which proves basin covariance. A homeomorphism commutes with taking boundaries, giving the third identity and the inclusions $G_A\leq G_{B(A)}\leq G_{\partial B(A)}$.

For $g\in G_{B(A)}$, choose $x_*$ with $\omega(x_*)=A$. Lemma~\ref{lem:basin} gives $x_*\in B(A)$, so $gx_*\in B(A)$ and
\[
gA=g\omega(x_*)=\omega(gx_*)\subseteq A.
\]
The same argument for $g^{-1}$ gives $A\subseteq gA$. Hence $gA=A$.
\end{proof}

\subsection{Admissible symmetry subgroups}
Fix a dynamical category $\mathcal C$, specifying the phase space, group action, regularity, choice of maps or flows, and an attractor class satisfying the above assumptions. As the dynamics vary within $\mathcal C$, let
\begin{align*}
\Adm_A(\mathcal C)&=\{G_A:A\text{ is an attractor in }\mathcal C\},\\
\Adm_B(\mathcal C)&=\{G_{B(A)}:A\text{ is an attractor in }\mathcal C\}.
\end{align*}
These are families of exact stabilizers.

\begin{corollary}\label{cor:admissibility}
For every such category,
\begin{equation}\label{eq:admissibility}
\Adm_A(\mathcal C)=\Adm_B(\mathcal C).
\end{equation}
\end{corollary}
\begin{proof}
Theorem~\ref{thm:transfer} identifies the two stabilizers for each realizing system and attractor.
\end{proof}

In particular, let $G\leq O(n)$ be finite and let $\mathcal R_G$ be its set of reflections. For $\Delta\leq G$, put
\[
L_\Delta=\bigcup_{\tau\in\mathcal R_G\setminus\Delta}\Fix(\tau).
\]
For a continuous $G$-equivariant map on $\R^n$, the classical restriction \cite{MDG1993,GS2002} states that an attractor stabilizer $\Sigma$ has a normal subgroup $\Delta\triangleleft\Sigma$ with
\begin{equation}\label{eq:GS}
\Sigma/\Delta\text{ cyclic},\qquad
\Delta C=C\text{ for a component }C\text{ of }\R^n\setminus L_\Delta.
\end{equation}
Corollary~\ref{cor:admissibility} transfers this necessary condition to $\Sigma=G_{B(A)}$, with the same setwise action on $C$.

For example, $D_3$ is excluded as an exact basin stabilizer for continuous $D_6$-equivariant maps in the standard planar representation \cite[Example~9.4]{GS2002}. We use $D_k$ for the dihedral group of order $2k$. The normal subgroups of $D_3$ are $\{e\}$, $C_3$, and $D_3$. The trivial subgroup gives a noncyclic quotient. Each remaining choice contains a $120^\circ$ rotation, which preserves none of the sectors of $\R^2\setminus L_\Delta$. Thus \eqref{eq:GS} excludes all three choices.

\section{Common boundaries and symmetry enhancement}\label{sec:boundary}
Fix $A$ and write
\[
\Gamma=\partial B(A),\qquad H=G_A=G_{B(A)},\qquad K=G_\Gamma.
\]
The attractor orbit $\mathcal O_A=\{gA:g\in G\}$ and the boundary orbit are identified with the left-coset spaces $G/H$ and $G/K$ \cite[Chapter~2]{Field2007}.

\begin{theorem}[Common-boundary multiplicity]\label{thm:fiber}
The assignment $gA\mapsto\partial B(gA)$ corresponds to the $G$-equivariant projection
\begin{equation}\label{eq:projection}
\pi:G/H\longrightarrow G/K,\qquad \pi(gH)=gK.
\end{equation}
Every fiber consists of exactly
\begin{equation}\label{eq:multiplicity}
m=[K:H]
\end{equation}
distinct symmetry-related attractors, whose basins have the same full boundary. In particular,
\begin{equation}\label{eq:strict}
K>H\quad\Longleftrightarrow\quad
\exists g\in G:\ gA\ne A\ \text{ and }\ \partial B(gA)=\partial B(A).
\end{equation}
\end{theorem}
\begin{proof}
The inclusion $H\leq K$ makes \eqref{eq:projection} well defined, and \eqref{eq:covariance} identifies it with the boundary assignment. The fiber over $K$ is $K/H$, so it contains $[K:H]$ elements. Left translation gives a bijection to every other fiber. Distinct attractors have disjoint basins by Lemma~\ref{lem:basin}. Finally, \eqref{eq:covariance} gives $g\in K$ exactly when $\partial B(gA)=\Gamma$, and $g\notin H$ exactly when $gA\ne A$.
\end{proof}

The subgroup $K$ acts transitively on the basins in one fiber. Its stabilizer at $B(A)$ is $H$. Thus boundary symmetry enhancement measures the number of attracting outcomes sharing that entire boundary within the specified symmetry orbit.

\subsection{Wada families and local separation}
Assume now that $A$ attracts an open neighborhood and $\Gamma\ne\varnothing$. All basins in $\mathcal O_A$ are open and pairwise disjoint. Following the common-boundary formulation of the Wada property \cite{KY1991,Daza2018}, a family of $q\geq3$ such basins is a \emph{full Wada family} when
\[
\partial B_1=\cdots=\partial B_q\ne\varnothing.
\]
Every point of the common boundary is a boundary point of each basin. The basins may have several connected components.

\begin{corollary}[Wada families and prime orbits]\label{cor:wada}
A boundary fiber of size $m\geq3$ is a full Wada family. If the attractor orbit has prime size $[G:H]=p$, then $K=H$ or $K=G$. For $p\geq3$, strict enhancement is equivalent to the full Wada property for that orbit.
\end{corollary}
\begin{proof}
The first assertion follows from Theorem~\ref{thm:fiber} and the definition. The index identity
\[
[G:H]=[G:K][K:H]
\]
leaves the two stated alternatives when $[G:H]$ is prime. The case $K=G$ places the entire orbit in one boundary fiber.
\end{proof}

\begin{proposition}[Local separation]\label{prop:local}
Let $B_1,\ldots,B_m$ be pairwise disjoint nonempty open sets with the same nonempty boundary $\Gamma$. For every $p\in\Gamma$ and every open neighborhood $U$ of $p$, the set $U\setminus\Gamma$ has at least $m$ connected components. On a manifold, a locally flat hypersurface patch of $\Gamma$ therefore requires $m\leq2$.
\end{proposition}
\begin{proof}
For any connected component $C$ of $U\setminus\Gamma$, the intersection $C\cap B_i$ is both relatively open and relatively closed in $C$: every limit point of $B_i$ outside $B_i$ belongs to $\partial B_i=\Gamma$. Hence $C\subseteq B_i$ whenever $C$ meets $B_i$. Each $B_i$ meets $U$, since $p\in\partial B_i$, and disjoint basins require distinct components. Near a locally flat hypersurface point, a suitable coordinate neighborhood has exactly two complementary components.
\end{proof}

\section{Examples}\label{sec:examples}
\subsection{A smooth common boundary with multiplicity two}
Set $s(u)=u(1-u^2)/(1+u^2)$ and consider
\begin{equation}\label{eq:double}
\dot x=s(x),\qquad \dot y=-y.
\end{equation}
The vector field is smooth with at most linear growth and generates a complete flow. It is equivariant under $G=\langle\kappa\rangle\cong C_2$, where $\kappa(x,y)=(-x,y)$. The equilibria $A_\pm=\{(\pm1,0)\}$ are asymptotically stable. A positive solution of the scalar equation converges monotonically to $1$; a negative solution converges to $-1$. Since $y(t)\to0$,
\[
B(A_+)=\{x>0\},\qquad B(A_-)=\{x<0\},\qquad
\partial B(A_\pm)=\{x=0\}.
\]
Consequently,
\begin{equation}\label{eq:double-groups}
G_{A_+}=G_{B(A_+)}=\{e\},\qquad
G_{\partial B(A_+)}=C_2,\qquad m=2.
\end{equation}
This gives strict enhancement for a straight boundary separating two basins.

\subsection{Four basins with distinct individual boundaries}
For the complete flow
\begin{equation}\label{eq:four}
\dot x=s(x),\qquad \dot y=s(y),
\end{equation}
the signed coordinate permutations form $G=D_4$. The equilibria $(\pm1,\pm1)$ attract their respective quadrants. For $A=\{(1,1)\}$,
\[
B(A)=\{x>0,y>0\},\qquad
\Gamma=\{(x,0):x\geq0\}\cup\{(0,y):y\geq0\}.
\]
The identity and $\sigma(x,y)=(y,x)$ are the only elements preserving either $A$ or $\Gamma$. Hence
\begin{equation}\label{eq:four-groups}
G_A=G_{B(A)}=G_\Gamma=\langle\sigma\rangle\cong C_2,\qquad m=1.
\end{equation}
The union of all four basin boundaries is the pair of coordinate axes and has full $D_4$ symmetry. Each individual boundary consists of two adjacent half-axes, as illustrated in Figure~\ref{fig:basins}(a).

\subsection{Newton maps with arbitrary multiplicity}
Consider the rational map on the Riemann sphere $\widehat\C$,
\begin{equation}\label{eq:newton}
N_d(z)=z-\frac{z^d-1}{dz^{d-1}}
=\frac{(d-1)z^d+1}{dz^{d-1}},\qquad d\geq3.
\end{equation}
The roots $\alpha_j=\exp(2\pi i j/d)$ are superattracting fixed points. We use their full basins, including all preimages of the immediate attracting components.

For a rational map $R$ of degree at least two, the full basin $U$ of an attracting fixed point satisfies $\partial U=J(R)$ \cite[Sections~2--3 and~6]{Milnor1991}. The basin $U$ lies in the Fatou set, is open, and is completely invariant. On a connected Fatou neighborhood meeting $U$, normality and the identity theorem force the iterates to converge to that fixed point, so $\partial U\subseteq J(R)$. If a neighborhood of a Julia point avoided $U$, all iterates on it would omit three chosen distinct points of $U$. Montel's theorem would give normality there. Thus $J(R)\subseteq\overline U$, proving the equality.

The isometries $r(z)=\alpha_1z$ and $\kappa(z)=\bar z$ of the spherical metric commute with $N_d$ and generate $G=D_d$. For $A=\{1\}$, the root stabilizer is $H=\langle\kappa\rangle$. The full root basins share $J(N_d)$, giving
\begin{equation}\label{eq:newton-groups}
G_A=G_{B(A)}=\langle\kappa\rangle,\qquad
G_{\partial B(A)}=D_d,\qquad m=d.
\end{equation}
The Newton family therefore realizes every common-boundary multiplicity $d\geq3$. Its Wada basins have also been used as numerical examples in \cite{Daza2018}.

Figure~\ref{fig:basins}(b) shows $N_3$ on an $800\times800$ cell-centered grid in $[-2,2]^2$. The iteration stops within $10^{-10}$ of a root or after $120$ iterations; an iterate with $|z|<10^{-14}$ is left unclassified. Every grid point resolved in this computation. The four-well labels in panel~(a) are assigned from the quadrant formulas.

\begin{figure}[tbp]
\centering
\begin{minipage}{.48\textwidth}
\centering
\includegraphics[width=\linewidth]{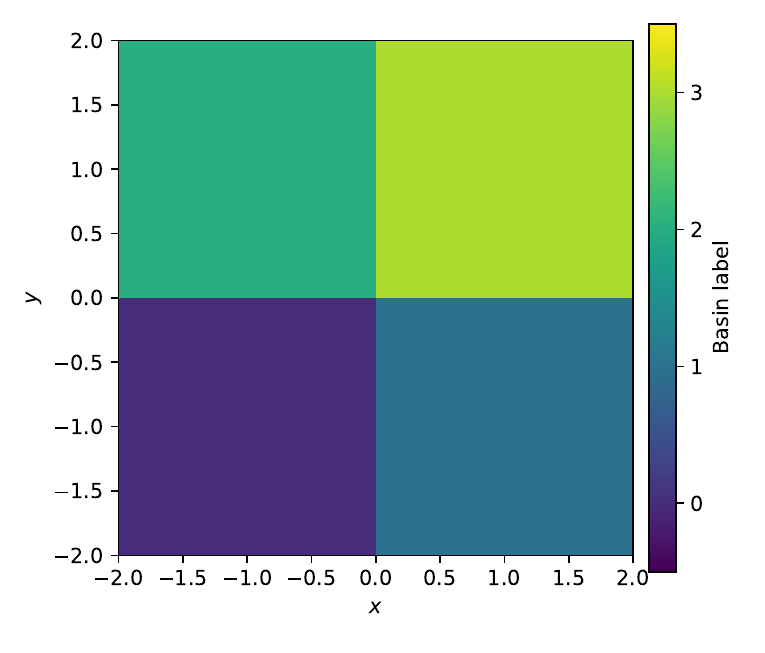}\\[-2pt]
\small (a) Four-well flow: $m=1$
\end{minipage}\hfill
\begin{minipage}{.48\textwidth}
\centering
\includegraphics[width=\linewidth]{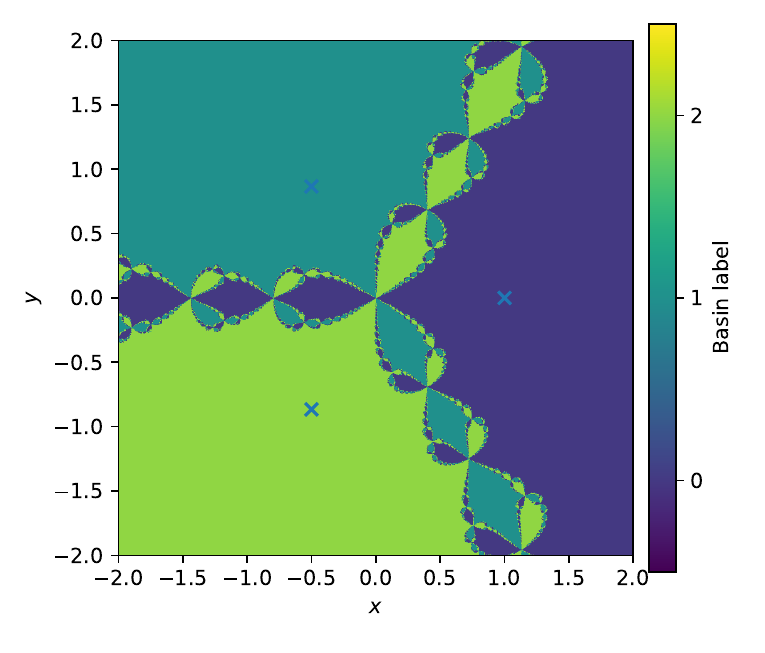}\\[-2pt]
\small (b) Cubic Newton map: $m=3$
\end{minipage}
\caption{Basin labels in two examples. In (a), four distinct individual boundaries form a $D_4$-invariant union. In (b), all three full root basins have the Julia set as their common boundary. Crosses mark the three roots.}
\label{fig:basins}
\end{figure}

Common-boundary geometry also occurs with trivial commuting symmetry. Take $p(z)=z(z-1)(z-2-i)$ and $R(z)=z-p(z)/p'(z)$. Its three root basins share $J(R)$. Any Euclidean isometry commuting with $R$ must permute its finite fixed points, which are precisely the roots. Their triangle has side lengths $1$, $\sqrt2$, and $\sqrt5$, so only the identity preserves it. Relative to this commuting group, $H=K=\{e\}$, while the three basins form a full Wada family.

\section{Conclusion and outlook}
The basin of a stable $\omega$-realized attractor inherits its exact setwise stabilizer and its admissibility restrictions. Boundary symmetry is governed by common-boundary multiplicity: the index $[G_{\partial B(A)}:G_A]$ counts the symmetry-related basins sharing that boundary. The explicit examples exhibit a smooth two-basin enhancement, distinct individual boundaries with a symmetric union, and Wada families of arbitrary multiplicity.

A further problem is to characterize the pairs of stabilizers $G_A\leq G_{\partial B(A)}$ realizable in a fixed class of flows or maps. For complete flows, basin boundaries are invariant. Equivariant stable-manifold theory for normally hyperbolic relative equilibria and relative periodic orbits \cite[Section~8.2.4]{Field2007} provides tools for studying boundaries organized by these objects. A boundary classification also requires a description of their global stable sets and the resulting separation of basins. A concrete computational direction is a full-boundary study of the Thomas flow \cite{Thomas1999}.

\section*{Data and code availability}
The source package includes the basin-label arrays and Python code for Figure~\ref{fig:basins}.

\section*{Acknowledgments}
The author thanks Professor Guanrong Chen and Dr.\ Jingxi Yang for discussions concerning symmetry and basin geometry.


\begin{thebibliography}{11}
\bibitem{GS2002}
M. Golubitsky and I. Stewart,
\emph{The Symmetry Perspective: From Equilibrium to Chaos in Phase Space and Physical Space},
Progress in Mathematics, vol.~200, Birkh\"auser, Basel, 2002.

\bibitem{Field2007}
M. Field,
\emph{Dynamics and Symmetry},
ICP Advanced Texts in Mathematics, vol.~3, Imperial College Press, London, 2007.

\bibitem{Li2017}
C. Li, J.~C. Sprott, and H. Xing,
Constructing chaotic systems with conditional symmetry,
\emph{Nonlinear Dynamics} \textbf{87} (2017), 1351--1358.

\bibitem{Li2023}
C. Li, Z. Li, Y. Jiang, T. Lei, and X. Wang,
Symmetric strange attractors: A review of symmetry and conditional symmetry,
\emph{Symmetry} \textbf{15} (2023), 1564.

\bibitem{MDG1993}
I. Melbourne, M. Dellnitz, and M. Golubitsky,
The structure of symmetric attractors,
\emph{Archive for Rational Mechanics and Analysis} \textbf{123} (1993), 75--98.

\bibitem{AM1994}
P. Ashwin and I. Melbourne,
Symmetry groups of attractors,
\emph{Archive for Rational Mechanics and Analysis} \textbf{126} (1994), 59--78.

\bibitem{FMN1996}
M. Field, I. Melbourne, and M. Nicol,
Symmetric attractors for diffeomorphisms and flows,
\emph{Proceedings of the London Mathematical Society} \textbf{72} (1996), 657--696.

\bibitem{KY1991}
J. Kennedy and J.~A. Yorke,
Basins of Wada,
\emph{Physica D} \textbf{51} (1991), 213--225.

\bibitem{Daza2018}
A. Daza, A. Wagemakers, and M.~A.~F. Sanju\'an,
Ascertaining when a basin is Wada: the merging method,
\emph{Scientific Reports} \textbf{8} (2018), 9954.

\bibitem{Milnor1991}
J. Milnor,
\emph{Dynamics in One Complex Variable: Introductory Lectures},
Stony Brook IMS Preprint 1990/5, revised September 1991; arXiv:math/9201272.

\bibitem{Thomas1999}
R. Thomas,
Deterministic chaos seen in terms of feedback circuits: Analysis, synthesis,
``labyrinth chaos'',
\emph{International Journal of Bifurcation and Chaos} \textbf{9} (1999), 1889--1905.
\end{thebibliography}
\end{document}